\documentclass[11pt]{amsart}
\usepackage[top=1in,bottom=1in,left=1in,right=1in]{geometry}
\usepackage{amsmath}
\usepackage{amsfonts}
\usepackage{amsthm}
\usepackage[foot]{amsaddr}
\usepackage{mathrsfs}  
\usepackage{tikz}
\usepackage{subcaption}
\usepackage{algorithm,algorithmic}

\newtheorem{theorem}{Theorem}

\newtheorem{lemma}{Lemma}

\title{The Arithmetic Circuit Combinatorial Nullstellensatz is NP-hard}

\author{Andreas Bj\"{o}rklund, Sweden}

\begin{document}
 \maketitle
   
\begin{abstract}
A multivariate polynomial on $n$ variables $x_1,\ldots,x_n$ of total degree $n$ over $\mathbf{Z}_2$ containing the multilinear monomial $\prod_{i=1}^n x_i$ is by the combinatorial nullstellensatz [Alon, Comb. Probab. Comput., 1999] known to always have a nonroot.
We show that there cannot be a randomised polynomial time algorithm that given an arithmetic circuit of polynomial size formally computing such a polynomial, locates a nonroot with constant nonzero probability unless RP=NP. The result holds even when the individual degree of every variable in the input polynomial is at most two. 
\end{abstract}

\section{Introduction}
The \emph{combinatorial nullstellensatz} by Alon~\cite{Alon1999} has found numerous applications in existence proofs in combinatorics, see e.g.~\cite{Alon2000, AlonBCMS2021}. It entails proving a given combinatorial structure exists by capturing its properties in a multivariate polynomial, and then arguing that a constrained nonroot to the polynomial always exists, thereby witnessing the structure.
 
\begin{theorem}[Second Combinatorial Nullstellensatz, Theorem 1.2 in Alon~\cite{Alon1999}]
\label{thm: null-general}
Let $\mathbf{F}_q$ be a finite field, and let $f=f(x_1,\ldots,x_n)$ be a polynomial in $\mathbf{F}_q[x_1,\ldots, x_n]$. Suppose the degree $\operatorname{deg}(f)$ of $f$ is $\sum_{i=1}^n t_i$, where each $t_i$ is a nonnegative integer, and suppose the coefficient of $\prod_{i=1}^n x_i^{t_i}$ in $f$ is nonzero. Then, if $S_1,\ldots,S_n$ are subsets of $\mathbf{F}_q$ with $|S_i|>t_i$, there are $a_1\in S_1,\ldots, a_n\in S_n$ so that
\[
f(a_1,\ldots, a_n)\neq 0.
\]
\end{theorem}

Hence, the computational search problem of locating a nonroot always has a solution. One fascinating aspect of Theorem~\ref{thm: null-general}, as pointed out in the original paper~\cite{Alon1999}, is that there is no known efficient algorithm for this task when the input polynomial is described implicitly, e.g. by a polynomial size arithmetic circuit.  Verifying a purported nonroot though, can be done in deterministic polynomial time by evaluating the polynomial through its input circuit description in the given point. Note that these problems (depending on the precise input specifications) are \emph{not} expected to be in the complexity class TFNP of total search problems in FNP, defined by Megiddo and Papadimitriou~\cite{MegiddoP1991}, as they are promise problems. The crux is that the input describes a polynomial that is guaranteed to contain a certain monomial, and verifying this is itself in general a hard problem (we discuss this further in the related work section~\ref{sec: related work}). 

In this paper, we focus on the above search problem when the input polynomial is given by an arithmetic circuit over $\mathbf{Z}_2$ computing the polynomial. 
Over $\mathbf{Z}_2$, the interesting version of theorem~\ref{thm: null-general} is
\begin{theorem}[Combinatorial Nullstellensatz over $\mathbf{Z}_2$]
\label{thm: null Z_2}
Let $f=f(x_1,\ldots,x_n)$ be a polynomial in $\mathbf{Z}_2[x_1,\ldots, x_n]$. Suppose the degree $\operatorname{deg}(f)$ of $f$ is $n$, and suppose the coefficient of $\prod_{i=1}^n x_i$ in $f$ is nonzero. Then, there exists $(a_1,\ldots,a_n)\in \{0,1\}^n$ so that
\[
f(a_1,\ldots, a_n)\neq 0.
\]
\end{theorem}
In fact, it is well known that the coefficient of the maximum multilinear monomial $\prod_{i=1}^n x_i$ is equal to the parity of the number of nonroots (see e.g.~Belovs \emph{et al.}~\cite{BelovsIQSY2017}, and our Lemma~\ref{lem: parity}). This connection implies the above theorem, because an odd number of nonroots means at least one.

We next describe how the input polynomial could be represented. An \emph{arithmetic circuit} $C$ over a ring $\mathbf{R}$ operates on a set of variables $x_1,\ldots,x_n$ and the scalars in $\mathbf{R}$, and computes a multivariate polynomial $f(x_1,\ldots,x_n)$.
It is a directed acyclic graph in which every node of indegree zero at the bottom is called an \emph{input gate} and is labelled by either a variable or a field scalar. Every other node has indegree two and is either a \emph{sum gate} ($+$), and computes the sum of its inputs, or a \emph{product gate} ($\times$), and computes the product of its inputs. A single node with outdegree zero at the top formally computes the polynomial by evaluating the gate nodes in the circuit bottom-up.
By the size of a circuit we mean the number of sum and product gate nodes in the directed acyclic graph. If we give unique labels to all inputs and all gates' output signals, we can describe the circuit simply as a list of gates, with for each gate its type $+/ \times$, its two input signal labels, and its output label.

There is already a strong conditional hardness result for the combinatorial nullstellensatz over $\mathbf{Z}_2$ Theorem~\ref{thm: null Z_2} in the literature. Alon~\cite{Alon2002} showed that a general polynomial time algorithm that takes as input a polynomial size arithmetic circuit computing a multilinear polynomial in the Boolean ring $\mathbf{Z}_2[x_1,\ldots,x_n]/\left(x_1^2+x_1,\ldots,x_n^2+x_n\right)$ containing the multilinear monomial $\prod_{i=1}^n x_i$, is unlikely. It would imply that there are no one-way permutations, one-to-one functions that are computationally easy to compute but hard to invert.

We provide another conditional hardness result, that in particular shows that an efficient algorithm for the arithmetic circuit combinatorial nullstellensatz implies efficient randomised algorithms for all problems in NP. Our contribution is two-fold:

\begin{enumerate}
\item[1]  Our input polynomial is given by a polynomial size arithmetic circuit \emph{formally} representing the polynomial in $\mathbf{Z}_2[x_1,\ldots,x_n]$. The difference to Alon's setting~(Section 5 in \cite{Alon2002}), is that there it is only required that the multilinear polynomial $f(x_1,\ldots,x_n)$ is \emph{pointwise} represented by an arithmetic circuit $C_f$ over $\mathbf{Z}_2$ on $\{0,1\}^n$, i.e., for all $(a_1,\ldots,a_n)\in \{0,1\}^n$, it holds that $f(a_1,\ldots,a_n)=C_f(a_1,\ldots, a_n)$. Indeed, Alon's input circuit is defined as computing a polynomial in the Boolean ring $\mathbf{Z}_2[x_1,\ldots,x_n]/\left(x_1^2+x_1,\ldots,x_n^2+x_n\right)$, whereas we define the circuit to compute a polynomial in the ring $\mathbf{Z}_2[x_1,\ldots,x_n]$. The circuit $C_f$ in the construction in~\cite{Alon2002} without the repeated substitutions $x_i^2\rightarrow x_i$, formally computes a polynomial in $\mathbf{Z}_2[x_1,\ldots,x_n]$ that is not itself multilinear and is of higher degree than $n$. 

\item[2] We give a randomised polynomial time reduction that we call \emph{compact vanishing or odd}, that given an instance $\mathcal I$ to an NP-complete problem constructs an arithmetic circuit computing a multivariate polynomial $f$ with $n$ variables $x_1,\ldots,x_n$, with $n$ a polynomial in the instance size, with three properties: 

\begin{enumerate}
\item[i] (compactness) $f$ has total degree $n$ and individual degree at most two of each variable,
\item[ii] (vanishing negatives) when $\mathcal I$ is a ``no''-instance, $f$ evaluates to zero everywhere on $\{0,1\}^n$, 
\item[iii] (odd positives) when $\mathcal I$ is a ``yes''-instance, $f$ has an odd number of nonroots with probability at least $1/\operatorname{poly}(n)$.

\end{enumerate}
\end{enumerate}

We prove 

\begin{theorem}[Main result: Hardness for the arithmetic circuit combinatorial nullstellensatz]
\label{thm: main}
Let $f=f(x_1,\ldots,x_n)$ be a polynomial in $\mathbf{Z}_{2}[x_1,\ldots, x_n]$ and let $C_f$ be an arithmetic circuit of polynomial size in $n$ that formally computes $f$. Suppose the total degree $\operatorname{deg}(f)=n$, the individual degrees of all $x_i$ are at most two, and suppose the coefficient of $\prod_{i=1}^n x_i$ in $f$ is nonzero. If there is a, possibly randomised, algorithm that given $C_f$ finds $(a_1,\ldots, a_n)\in \{0,1\}^n$, in time $\operatorname{poly}(n)$ and with nonzero constant probability, so that
\[
f(a_1,\ldots, a_n)\neq 0,
\]
then RP=NP.
\end{theorem}

Note that we make no assumption on what the hypothetical algorithm would output when presented with an arithmetic circuit that computes a polynomial \emph{not} containing the maximum degree multilinear monomial $\prod_{i=1}^n x_i$. It may very well produce a nonroot even when this monomial is lacking. 
Moreover, the individual degree bound of two is best possible since for nonzero multilinear polynomials, efficient algorithms locating a nonroot exists, see the next section. 
 
\subsection{Related work}
\label{sec: related work}
The complexity class TFNP (Total Functions from NP) was introduced by Megiddo and Papadimitriou~\cite{MegiddoP1991}. It is defined as follows: Assume an alphabet $\Sigma$ with at least two symbols, and a relation $R\subseteq \Sigma^* \times \Sigma^*$ that is polynomially balanced ($(x,y)\in R$ implies $|y|\leq p(|x|)$ for some polynomial $p$), and polynomial time recognisable
(we can check if $(x,y)\in R$ in polynomial time). The relation $R$ is \emph{total} if for every $x\in \Sigma^*$, there is always a $y\in \Sigma^*$ such that $(x,y)\in R$. TFNP is the class of problems with such relations. Our search problem in Theorem~\ref{thm: main} is not known to be in this class as its input is assumed to contain the maximum multilinear monomial, which is hard to verify in general. Indeed, computing the coefficient of the multilinear monomial containing all variables of a polynomial over $\mathbf{Z}_2$ is $\oplus$P-hard (proved e.g. in Proposition 7, Belovs \emph{et al.}~\cite{BelovsIQSY2017}), which in particular means that there are no known efficiently verifiable proofs that can be guessed and validated by a nondeterministic polynomial time verifier.
  There does however exist a restricted version of our problem, called \textsc{PPA-Circuit CNSS}, (where CNSS stands for Combinatorial Nullstellensatz), which is PPA-complete and hence is in TFNP, see Belovs \emph{et al.}~\cite{BelovsIQSY2017}.
In this version of the problem, the input polynomial is implicitly given as an input to an arbitrary PPA-problem, as a so-called PPA-circuit composition of the edge identification circuit and the pairing circuit of the PPA problem. However, it does not as in our case allow arbitrary polynomial size arithmetic circuits describing a degree $n$ polynomial as inputs.

Our main result Theorem~\ref{thm: main} can be contrasted by the fact that there is a simple and efficient algorithm finding nonroots in multilinear polynomials. Aggrawal and Thierauf~\cite{AggrawalT2000}, implicitly showed that if we are given an arithmetic circuit of polynomial size that formally computes a nonzero multilinear polynomial $f$ in $\mathbf{Z}_{2}[x_1,\ldots, x_n]$, we can find a nonroot $(a_1,\ldots, a_n)\in \{0,1\}^n, f(a_1,\ldots,a_n)\neq 0$ in randomised polynomial time. The idea is to view the arithmetic circuit as computing the polynomial over an \emph{extension field} $\mathbf{Z}_{2^\kappa}$ of $\mathbf{Z}_2$, e.g., by setting $\kappa=2\lceil \log_2 n \rceil$ (they attribute the general idea of using an extension field to Ibarra and Moran~\cite{IbarraM1983}). Then, by the DeMillo-Lipton-Schwartz-Zippel lemma~\cite{DeMilloL1978,Schwartz1980,Zippel1979}, the polynomial is nonzero in a uniformly sampled point in $\mathbf{Z}_{2^\kappa}^n$ with probability at least $1-\frac{1}{n}$. 

The algorithm is as follows: By uniform random sampling, we find a nonroot $a\in\mathbf{Z}_{2^\kappa}^n$ in the extension field.
Using self reduction,  we next turn the found nonroot $a$ into one that lies in $\{0,1\}^n$, by gradually changing the assignment one variable at a time. Formally, we will create a sequence of nonroots $a=a^{(0)},\ldots,a^{(\ell)}\in \mathbf{Z}_{2^\kappa}^n$, in which each nonroot differs in value from the previous one in exactly one variable, and the new variable value is contained in $\mathbf{Z}_2$ whereas the old value was not. Eventually, all variable values in the assignment lie in $\{0,1\}$.

To see in detail how this works, consider step $i$ in the sequence creation, and let $j$ be the smallest index such that $a^{(i)}_j\in \mathbf{Z}_{2^\kappa}\setminus \mathbf{Z}_2$. Rewrite the multivariate polynomial as a single variable polynomial in $x_j$ as 
\[
f_j(x_j)=c_0+c_1x_j,
\]
where $c_k$ are the coefficients obtained be replacing $x_l,l\neq j$ for $a^{(i)}_l$ in $f(x)$. Note that we have $f_j (x_j)\neq 0$, since $a^{(i)}$ is a nonroot to $f(x)$ in $\mathbf{Z}_{2^\kappa}$. This means some $c_k$ must be nonzero, and $f_j(x_j)$ is a nonzero linear polynomial that has at most one root. This in turn means that for at least some $g\in \{0,1\}$, $f_j(g)\neq 0$. Set $a^{(i+1)}_l=a^{(i)}_l$ for $l\neq j$, and $a^{(i+1)}_j=g$ for the smallest such $g$, to obtain a new nonroot $a^{(i+1)}$ with one less variable value outside $\mathbf{Z}_2$. Note that to find the value $g$ at each step, we simply have to loop over $g\in \{0,1\}$ and evaluate $f_j(g)$ with the help of the arithmetic circuit.

\section{The high-level NP-hardness of the arithmetic circuit combinatorial nullstellensatz}
\label{sec: high-level proof}
In this section we give the high-level proof Theorem~\ref{thm: main}. Our proof strategy is to describe an embedding of an arbitrary instance $\mathcal{I}$ to a NP-complete problem in a probabilistic arithmetic circuit $C_f$. The circuit computes a multivariate polynomial $f(x_1,\ldots,x_n)$ in $\mathbf{Z}_2[x_1,\ldots,x_n]$, where $n$ is polynomial in the size of $\mathcal{I}$, with three properties.

\vspace{5mm}
\begin{enumerate}
\item[P1] (Compactness) The degree $\operatorname{deg}(f)$ is at most $n$ and the individual degrees are at most two, 
\item[P2] (Vanishing negatives) when $\mathcal{I}$ is a ``no''-instance, $f$ has no nonroots,
\item[P3] (Odd positives) when $\mathcal{I}$ is a ``yes''-instance, with probability at least $1/\operatorname{poly}(n)$ the polynomial has an odd number of nonroots.
\end{enumerate}
\vspace{5mm}
The circuit $C_f$ will be of size polynomial in $n$, and the reduction runs in time polynomial in $n$. We call such a construction a \emph{compact vanishing or odd} reduction.

Before we describe such a reduction in Section~\ref{sec: cvo}, we first give the high-level proof of Theorem~\ref{thm: main} assuming one exists. We describe how we given a compact vanishing or odd reduction can use the hypothetical algorithm in Theorem~\ref{thm: main}, in the following called algorithm $\mathcal{A}$, to construct a polynomial time Monte Carlo algorithm without false positives for any problem in NP.

We first establish what the coefficient of the full multilinear monomial $\prod_{i=1}^n x_i$ counts in the polynomial $f$. Let $\mathcal P_n$ be the set of polynomials $f\in \mathbf{Z}_2[x_1,\ldots,x_n]$ of degree $\operatorname{deg}(f)=n$. 
The following fact relating the coefficient of the monomial $\prod_{i=1}^n x_i$ and the parity of the number of nonroots in a polynomial in $\mathcal P_n$ is well known (see e.g.~\cite{BelovsIQSY2017}):
\begin{lemma}
\label{lem: parity}
The coefficient of the monomial $\prod_{i=1}^n x_i$ in a polynomial $f\in \mathcal{P}_n$ is equal to the parity of the number of nonroots $(a_1,\ldots,a_n)\in\{0,1\}^n,f(a_1,\ldots,a_n)\neq 0$.
\end{lemma}
\begin{proof}
Consider the unique reduced polynomial $f^*$ in $\mathbf{Z}_2[x_1,\ldots,x_n]/(x_1^2+x_1,\ldots, x_n^2+x_n)$ that is pointwise identical to $f$ on $\{0,1\}^n$.
 
Next, consider $X\subseteq [n]$ and let $m_{f^*}(X)$ be the indicator function (coefficient) of the monomial $\prod_{i\in X} x_i$ in $f^*$.
Let $v_{f^*}(X)$ be the value $f^*(a_1,\ldots,a_n)$ in the point $a_i=1$ iff $i\in X$ for all $i=1,\ldots,n$.
Then, we have
\[
v_{f^*}(X)=\sum_{Y\subseteq X} m_{f^*}(Y),
\]
since the value of the polynomial in the point $X$ equals the number of monomials with coefficient one that are one under $X$.
M\"obius inversion on the subset lattice gives the formula
\[
m_{f^*}(X)=\sum_{Y\subseteq X} (-1)^{|X\setminus Y|} v_{f^*}(Y),
\]
from which we see that $m_{f^*}(\{1,\ldots,n\})$, the coefficient of the monomial $\prod_{i=1}^n x_i$ in $f^*$, is equal to the parity of the number of nonroots to $f^*$ and hence also to $f$ since they are pointwise identical on $\{0,1\}^n$. Since the degree of the polynomial is $n$, we have that the coefficient of the monomial $\prod_{i=1}^n x_i$ is the same in $f^*$ and $f$, and the result follows.
\end{proof}
Hence, when P3 applies and we have an odd number of nonroots, the multivariate monomial $\prod_{i=1}^n x_i$ will have coefficient one in the polynomial $f$ as assumed in Theorem~\ref{thm: main}.
\subsection{Overall Cook reduction}
\label{sec: overall}
First note that any problem in NP can be reduced to an NP-complete problem of our choice using known Karp reductions. Let an instance $\mathcal{I}$ to our chosen NP-complete problem be given. 
We next compute the associated probabilistic arithmetic circuit $C_f$ from $\mathcal{I}$ via the compact vanish or odd reduction and feed $C_f$ to algorithm $\mathcal{A}$. We check in polynomial time using $C_f$ that $\mathcal{A}$'s output is indeed a nonroot. If it found a nonroot, we will know from property P2 that $\mathcal{I}$ is a ``yes''-instance. If  algorithm $\mathcal{A}$ did not find a nonroot, we guess our original instance $\mathcal{I}$ was a ``no''-instance.
Note that there are no false positives and the probability of false negatives depends on two events, the event that the probabilistic arithmetic circuit $C_f$ contains the monomial $\prod_{i=1}^n x_i$, as well as the event of success of algorithm $\mathcal{A}$ in Theorem~\ref{thm: main}. The first event happens with probability that is inversely polynomial in $n$ according to property P3 and Lemma~\ref{lem: parity}. Since the algorithm $\mathcal{A}$'s success is independent of this event, the total probability is also inversely polynomial in $n$. Note that by repeating the whole process a large enough polynomial number of times in independent tries, answering affirmative if at least some try finds a nonroot, and negative otherwise, the probability of false negatives can be brought down to $\exp(-n^{O(1)})$.

\subsection{Previous vanishing or odd constructions}
\label{sec: prev vanishing}
It is not too difficult to find reductions that have two out of the three required properties of a compact vanishing or odd reduction. E.g., let us consider two results from the literature to derive a (noncompact) vanishing or odd reduction for a better understanding of where the difficulties lie in getting all three properties.

As our first example, perhaps the most natural approach is to consider the famous construction in Valiant and Vazirani~\cite{ValiantV1986}. It takes an instance $\phi$ to \textsc{CNF3Sat} and produces $n+2$ \textsc{CNF3Sat} instances $\phi_i,1\leq i \leq n+2$ with the property that if $\phi$ is a ``yes''-instance, some $\phi_i$ will have a unique satisfying assignment with probability at least $\frac{1}{8}$. If $\phi$ is a ``no''-instance though, all $\phi_i$ will be unsatisfiable. We can encode the satisfiability question of any of these instances $\phi_i$ as a question of nonroots to a polynomial with nonroots representing satisfiable assignments. Note that such a polynomial for a randomly chosen $\phi_i$ has exactly our vanishing or odd property.  We may use variables $x_1,\ldots,x_n\in \{0,1\}$ to represent the Boolean variables $v_1,\ldots,v_n$ in the instance with $x_1=1$ meaning $v_i$ is true, and use polynomials of degree three to encode the clauses. E.g., the polynomial $1+x_1(1+x_2)x_3$ in $\mathbf{Z}_2[x_1,x_2,x_3]$ is nonzero unless $x_1=1,x_2=0,x_3=1$, corresponding to a Boolean clause $\overline{v_1}\vee v_2 \vee \overline{v_3}$ which is false only if $v_1=T,v_2=F,v_3=T$.
Taking the product of all clauses' polynomials, we obtain an arithmetic circuit of a polynomial in $\mathbf{Z}_2[x_1,\ldots, x_n]$ that vanishes on $\{0,1\}^n$ when there is no satisfying assignment to the original \textsc{CNF3Sat} instance $\phi$. Moreover, when $\phi$ is satisfiable, with probability at least $1/\operatorname{poly}(n)$ the constructed polynomial has exactly one nonroot.  However, the total degree can be three times the number of clauses, which is much larger than $n$. Also, the individual degree is larger than two for most of the variables as NP-hard \textsc{CNF3Sat} instances typically need at least three occurrences of the variables.

As our second example, we will retrace a carefully tailored NP-completeness proof for a similar problem.  Buss, Frandsen, and Shallit~\cite{BussFS1997} showed that the problem of deciding if a mixed matrix (a matrix with scalars and variables) allows an assignment to the variables so that the matrix gets full rank is NP-complete.
It can be formulated as determining if a determinant polynomial
\[
f(x_1,\ldots,x_n)=\operatorname{det}\left(A_0+\sum_{i=1}^n x_iA_i\right),
\]
in $\mathbf{Z}_2[x_1,\ldots, x_n]$, where each $A_i$ is an $m\times m$ $0/1$ matrix, has a nonroot $(a_1,\ldots, a_n)\in \{0,1\}^n$, $f(a_1,\ldots,a_n)\neq 0$.  Furthermore, Harvey, Karger, and Yekhanin~\cite{HarveyKY2006} in their Corollary $3$ proved that this holds even if each $A_i$ for $i>0$ has at most two ones each. While this means that the determinant seen as a polynomial has individual degree of every variable at most two, the total degree of the polynomial is always larger than $n$ in these NP-hardness reductions. To see why we will formulate a version of the reduction idea in~\cite{HarveyKY2006} but expressed directly in terms of multivariate polynomials as opposed to ranks of matrices used in their paper.

Their reduction is made in two steps. The first step does not pay any respect to the individual degree of the variables, whereas the second step is a general method to decrease the individual degree to two in a way that does not increase the total degree of the polynomial more than it increases the number of variables.

The first step of the reduction can be seen as taking as input an arithmetic circuit $C_f$ over $\mathbf{Z}_2$ of a function $f$ on $n$ variables $x_1,\ldots,x_n$ and constructing another circuit $C_{f'}$ over $\mathbf{Z}_2$ of a function $f'$ on $n'$ variables $y_1,\ldots,y_{n'}$, with $n'>n$. The input $C_f$ could e.g. be the output arithmetic circuit of the previous reduction encapsulating the satisfiability of a \textsc{CNF3Sat} formula\footnote{Harvey, Karger, and Yehkanin~\cite{HarveyKY2006} reduced from the problem \textsc{CircuitSAT}.}.
The variables $y_1,\ldots,y_{n'}$ include the original variables, i.e., we can equate $y_i=x_i$ for $i=1,\ldots,n$, and there is one extra variable $y_i,i>n$ per gate in $C_f$ representing the output of the gate. For a multiplication gate $y_c=y_a\times y_b$, we associate the polynomial formula $(y_c+y_ay_b+1)$. Note that it will be $1$ iff $y_c=y_a\times y_b$. Similarly, for an addition gate $y_c=y_a+y_b$, we associate the polynomial formula $(y_c+y_a+y_b+1)$. Again it is one iff $y_c=y_a+y_b$. Now, we simply let $C_{f'}$ be the product of all these associated polynomial formulas for the gates, finally multiplied with $y_s$, the output gate variable of the circuit $C_f$ to ensure it is set to $1$. We now have that the nonroots to $f'$ correspond one-to-one with the nonroots to $f$. This means that $f'$ has both the vanishing and odd properties iff $f$ has.

To analyse the total degree of the construction, let the size of $C_f$ be expressed as $s(C_f)=s^\times(C_f)+s^+(C_f)$, where $s^\times(C_f)$ is the number of multiplication gates in $C_f$ and $s^+(C_f)$ the number of addition gates. The constructed circuit $C_{f'}$ will have
\[
n'=n+s(C_f),
\]
variables and
\[
\operatorname{deg}(f')\leq 2s^\times(C_f)+s^+(C_f)+1.
\]
Hence, to get also the compactness property for $C_{f'}$, i.e. $\operatorname{deg}(f')=n'$, we either need cancellation of high degree monomials to get strict inequality in the degree bound or $s^\times(C_f)=n-1$. Note that circuits with $s^\times(C_f)=n-1$ are quite limited. Constructing a compact vanishing or odd reduction via this circuit embedding technique alone appears challenging.

The second step of the reduction is interesting on its own. It takes any variable $y_c$ of degree $d(c)>2$ and makes as many copies $y_{c,1},\ldots,y_{c,d(c)}$ as the individual degree of the variable $y_c$ in the circuit $C_{f'}$ (this creates $d(c)-1$ new variables). Every occurrence of $y_c$ in the circuit is mapped to its own copied version $y_{c,i}$ of the variables. To make sure all copies take the same value, we form the polynomial $\prod_{i=1}^{d(c)} y_{c,i}+ \prod_{i=1}^{d(c)} (1+y_{c,i})$. It is of degree $d(c)-1$, and is one iff all $y_{c,i}, i=1,\ldots, d(c)$ are equal. By multiplying this polynomial to the original polynomial $f'$, we get a circuit computing a polynomial with more variables and higher degree, but the increase in the total degree and the number of variables are by the same amount. Hence this reduction can be applied to any construction of a polynomial with degree equal to its number of variables while retaining this property.

\section{A compact vanishing or odd reduction}
\label{sec: cvo}
In this section we describe a compact vanishing or odd reduction, which together with the high level proof in Section~\ref{sec: overall} proves our main Theorem~\ref{thm: main}.
Our hardness reduction is from the problem of deciding if a cubic bipartite planar undirected graph has a Hamiltonian cycle, a problem that was proved NP-complete by Akiyama, Nishizeki, and Saito~\cite{AkiyamaNS1980}. For technical reasons that will become clear in what follows, we look for a specific type of Hamiltonian cycle. We require both that it goes through a specific edge $e^*$, and we also give small positive integer weights to the edges of the graph and seek a Hamiltonian cycle of a specific total weight $w^*$ (sum of the edge weights along the cycle). We say a Hamiltonian cycle is \emph{special} if it meets both requirements.

\subsection{Directed cycle covers in bipartite planar graphs}
We are interested in the Hamiltonian cycles through a specified edge in the graph, but it does not matter which one, we can just consider a randomly chosen edge $e^*$. With probability at least $\frac{2}{3}$ it will be part of any Hamiltonian cycle when one exists in the graph since the graph is cubic. To transform the Hamiltonicity detection problem into a question about nonroots of a polynomial, we connect it to an auxiliary problem introduced by Bj\"orklund, Kaski, and Nederlof~\cite{BjorklundKN2024}.  We call it \textsc{Realisable Cycle Cover} here. Before we describe it and prove its relationship to Hamiltonian cycles in a bipartite planar graph in Section~\ref{sec: vccc}, we set up some terminology.

\subsubsection{Graph terminology}
We consider both undirected and directed graphs. For a loopless graph $G=(V,E)$, we let $\{u,v\}$ denote an undirected edge between vertices $u$ and $v$ when $G$ is undirected, whereas $(u,v)$ denotes the directed edge, or \emph{arc} from $u$ to $v$ when $G$ is directed. 
A \emph{directed cycle} $C$ in a graph is a sequence of distinct vertices $(c_0,c_1,\ldots,c_{\ell-1})$ such that $\{c_i,c_{(i+1)\mbox{\tiny mod }\ell}\}\in E$ if the graph is undirected, or $(c_i,c_{(i+1)\mbox{\tiny mod }\ell})\in E$ if the graph is directed.
We write $V(C)$ to refer to the set of vertices $\{c_0,\ldots,c_{\ell-1}\}$ on the cycle.  
A \emph{directed cycle cover} in a graph is a set of vertex disjoint directed cycles such that every vertex is part of some cycle.

\subsubsection{Pfaffian orientations}
In an undirected graph $G=(V,E)$, a \emph{Pfaffian orientation} of $G$ is an orientation $P$ of the edges $\{u,v\}\in E$, i.e., precisely one of the directed arcs $(u,v)$ or $(v,u)$ is in $P$, with a certain parity property: along any directed cycle $C$ in $G$ such that $G[V\setminus V(C)]$ admits a perfect matching, there is an odd number of edges oriented in either of the two directions. Such a cycle $C$ is called \emph{central}. A graph that admits a Pfaffian orientation is called a \emph{Pfaffian} graph. Here we will only use that planar graphs are Pfaffian, and that it is computationally easy to find a Pfaffian orientation in a planar graph, see e.g. Little~\cite{Little1974} for a polynomial time algorithm.
 
 \subsubsection{Realisable cycle covers}
 \label{sec: vccc}
Let $P$ be a Pfaffian orientation of a bipartite undirected graph $G=(V,E)$.
Fix an arc $(s,t)$ in $P$ and let $P_{s,t}$ be obtained from $P$ by reversing the arc $(s,t)$ to $(t,s)$. Finally, let $p_{s,t}:E\rightarrow\{0,1\}$ be the indicator function for an arc being part of $P_{s,t}$. Note that precisely one of $p_{s,t}((u,v))$ and $p_{s,t}((v,u))$ is $1$ for every $\{u,v\}\in E$, and all other points evaluate to zero. 

The key properties of the tuple $(G,P_{s,t})$ that we will describe next were observed in \cite{BjorklundKN2024}. We number the vertices $V(G)$ as $1,\ldots,n$, and we introduce binary variables $x_i$ for all $i=1,\ldots, n$. We say an arc $(u,v)$ in a graph $G$ under $P_{s,t}$ is \emph{realised} by a variable assignment when
\begin{equation}
\label{eq: realisable arc}
x_{u}+x_{v}=p_{s,t}\left((u,v)\right)\mbox{ mod }2.
\end{equation}
A directed cycle $C=(c_0,\ldots,c_{\ell-1})$ is \emph{realisable} in $G$ under $P_{s,t}$ if there exists a variable assignment realising all arcs $\forall i\in\{0,\ldots,\ell-1\}:(c_i,c_{(i+1)\mbox{\tiny mod } \ell})$.

\begin{lemma}
\label{lem: cycle without e*}
Every central cycle that is realisable in $G$ under $P_{s,t}$ must pass through $\{s,t\}$.
\end{lemma}
\begin{proof}
We sum the left and right sides of the equations in~\eqref{eq: realisable arc} for all arcs on the cycle to reach a contradiction.
First, consider the right side sum, $\sum_{i=0}^{\ell-1} p_{s,t}((c_i,c_{(i+1)\mbox{\tiny mod }\ell}))$. It is odd when $\{s,t\}$ is not on the central cycle, 
because of the property of a Pfaffian orientation $P$. Second, the left side sum, $\sum_{i=0}^{\ell-1} x_{c_i}+x_{c_{(i+1)\mbox{\tiny mod }\ell}}$ is always even, because each vertex variable is counted twice. Hence it is impossible to simultaneously realise all arcs on the cycle.
\end{proof}
We say a directed cycle cover is realisable in $G$ under $P_{s,t}$ iff all its cycles are.
\begin{lemma}
\label{lem: realisable}
A directed cycle cover is realisable in $G$ under $P_{s,t}$ iff it consists of a Hamiltonian cycle through $\{s,t\}$.
\end{lemma}
\begin{proof}
(if:) Given a directed Hamiltonian cycle $h=(s=h_0,\ldots,h_{n-1}=t)$ in $G$ through $\{s,t\}$, we can find an asisgnment $a_1,\ldots,a_n$ to $x_1,\ldots,x_n$ that realises all arcs on the cycle, by setting $a_s=0$, and $a_{h_{i+1}}=a_{h_i}+p((h_i,h_{i+1}))$ for $i=0,\ldots, n-2$. (The solution is not unique, by setting $a_s=1$ we obtain another one.)

(only if:) Every cycle in a directed cycle cover in a bipartite graph is central, as one can always form a perfect matching of the vertices not on the cycle by pairing all vertices by every other arc along the cycles. This pairing is always possible since all cycles have even length in a bipartite graph. From Lemma~\ref{lem: cycle without e*}, every cycle in a cycle cover must go through $\{s,t\}$. Since the cycles in a cycle cover are vertex and hence edge disjoint, the cycle cover must consist of a single cycle through $\{s,t\}$. 
\end{proof}

The problem \textsc{Realisable Cycle Cover} is given a bipartite planar graph $G$, a Pfaffian orientation $P$, and an edge $\{s,t\}$ with $(s,t)\in P$, to detect if $G$ has a realisable cycle cover under $P_{s,t}$.
Combining Lemma~\ref{lem: realisable} with the NP-hardness of detecting if a cubic bipartite planar graph has a Hamiltonian cycle~\cite{AkiyamaNS1980} we have the following: 

\begin{lemma}  \textsc{Realisable Cycle Cover} is NP-hard.
\end{lemma}

\subsection{Multivariate polynomial embedding}

We will next see how \textsc{Realisable Cycle Cover} can be encoded in a multivariate polynomial with nonroots corresponding to solutions. We will use a matrix determinant over the field $\mathbf{Z}_{2}$ to define the polynomial. Leibniz formula for the determinant of a square $n\times n$ matrix $M$ consider all possible cycle cover of a graph described by $M$:
\begin{equation}
\label{eq: Leibniz}
\operatorname{det}(M)=\sum_{\sigma}\operatorname{sgn}(\sigma)\prod_{i=1}^nM_{i,\sigma(i)},
\end{equation}
where the summation is over all permutations $\sigma$ of $1,\ldots,n$. Here the sign function $\operatorname{sgn}(\sigma)$ takes values $-1$ or $1$ depending on the sign of the permutation $\sigma$, but as we are working in $\mathbf{Z}_2$, these are all one. We think of the rows and columns as labelled by the vertices of $G$, and an entry $M_{u,v}$ is nonzero only if $\{u,v\}\in E(G)$, and it encodes the directed arc $(u,v)$. From the formula, we see that all summands can be viewed as directed cycle covers in $G$, where each cycle cover contributes the product of all the $M_{u,v}$ for each arc $(u,v)$ in the cycle cover.
 
We will also use an auxiliary variable $y$ that we need to track the total weight of the Hamiltonian cycles. Our final polynomial computed by the arithmetic circuit will \emph{not} contain the variable $y$ however, it is only used in the construction and analysis of the circuit. We assume there are small positive integer weights $w:E\rightarrow [1,2,\ldots,b]$ on each edge (the edge weights will be set below in property P3 where they are needed).

We next define the matrix. We introduce an asymmetry around the special edge $\{s,t\}$ to make sure that only certain variable assignments, those with $x_s=0$ and indirectly $x_t=1$, will possibly be nonroots to our polynomial. We will see why this matrix definition works in Sections~\ref{sec: P2} and~\ref{sec: P3}. 

\begin{equation}
\label{eq: M_e}
M_{s,t}(x)_{u,v}=\left\{\begin{array}{ll} y^{w(\{t,s\})}(x_s+1) & : u=t,v=s, \\
y^{w(\{u,v\})}(x_u+x_v+p_{s,t}((u,v))+1) & :\{u,v\}\in E\setminus \{s,t\},
 \\ 0 & :\mbox{otherwise.} \end{array} \right.
\end{equation}

We use the notation $[y^k]p(y)$ to refer to the coefficient of $y^k$ in a polynomial $p(y)$. Our polynomial $f$, computed by the circuit $C_f$ we will construct, will be

\begin{equation}
\label{eq: f}
f(x_1,\ldots,x_n)=[y^{w^*}]\operatorname{det}(M_{s,t}(x)).
\end{equation}

 We proceed to prove that $f$ as defined in \eqref{eq: f} has the three properties P1--P3 of a compact vanishing or odd reduction polynomial.
 
\subsubsection{Property P1, compactness}
\label{sec: P1}
We verify that the total degree $\operatorname{deg}(f)$ of $f$ in \eqref{eq: f} is at most $n$, and that the individual degree of any variable is at most two.
There are $n$ variables $x_1,\ldots, x_n$ in our $f$ and it is defined via the determinant of an $n\times n$ matrix $M(x)$ with linear polynomials on the form $c_0+\sum_i c_ix_i$ as entries, where $c_i$ are constants in the field. From the Leibniz formula~\eqref{eq: Leibniz} for the determinant, that computes a sum over products of one entry per row and column, it is clear that any monomial has degree at most $n$ in the $x_1,\ldots,x_n$ variables. This establishes $\operatorname{deg}(f)\leq n$.

Every nonzero summand in the Leibniz formula~\eqref{eq: Leibniz} for $\operatorname{det}(M_{s,t}(x))$ comes from a permutation $\sigma$ that describes a directed cycle cover in the underlying graph $G$. This means in particular that every vertex is part of exactly two arcs (the matrix $M_{s,t}(G)$ in \eqref{eq: M_e} is zero on the diagonal so there are no loops), and since every variable $x_u$ is only part of the matrix entries corresponding to arcs that are incident to vertex $u$, we know that no monomial in the polynomial can have an individual degree larger than two for any of the variables $x_1,\ldots, x_n$.

\subsubsection{Property P2, vanishing negatives}
\label{sec: P2}
We verify that $f$ in \eqref{eq: f} has a nonroot only if $G$ has a Hamiltonian cycle.

\begin{lemma}  
\label{lem: vanishing}
If there exists a nonroot $(a_1,\ldots,a_n)\in\{0,1\}^n$,$f(a_1,\ldots,a_n)\neq 0$, then $G$ has a Hamiltonian cycle through $e^*=\{s,t\}$.
\end{lemma}
\begin{proof}
The proof is seen via Leibniz formula for the determinant~\eqref{eq: Leibniz}. Every contributing permutation corresponds to a cycle cover in the original graph $G$, as only entries in the matrix $M_{s,t}(x)$ corresponding to arcs in the graph $G$ are nonzero. Furthermore, each nonzero entry in the matrix except the one for the arc $(t,s)$ encodes the formula \eqref{eq: realisable arc} for realisability, meaning that the entry is nonzero only if the corresponding arc is realised by the assignment $a_1,\ldots,a_n$. From Lemma~\ref{lem: cycle without e*} and Lemma~\ref{lem: realisable}, we see that there cannot be a cycle cover with a directed cycle not passing through $e^*$ regardless of the assignment to the variables $x_1,\ldots,x_n$. Hence, all contributing cycle covers in the determinant must be Hamiltonian cycles in $G$ through $e^*$.
\end{proof}

\subsubsection{Property P3, odd positives}
\label{sec: P3}
We verify that when $G$ has a Hamiltonian cycle, with probability at least $1/\operatorname{poly}(n)$ there is an odd number of nonroots.

We need a construction that makes it likely that a ``yes''-instance corresponds to a polynomial with an odd number of nonroots. One remaining problem towards accomplishing this, is that a cubic graph always has an even number of Hamiltonian cycles through every edge (Smith's theorem, see Tutte~\cite{Tutte1946}).  To circumvent this, we need to introduce some randomness to break the structure while retaining properties P1 and P2. One way is to use the well-known isolating lemma by Mulmuley, Vazirani, and Vazirani~\cite{MulmuleyVV1987}:

\begin{lemma}[Isolating Lemma, \cite{MulmuleyVV1987}]
Let $b$ be a positive integer, let $U$ be a finite universe, and let $\mathcal{F}$ be an arbitrary nonempty family of subsets of $U$.
Suppose each element $x\in U$ gets an independent and uniformly sampled weight $w(x)\in \{1,2,\ldots,b\}$, and define the weight of $f\in \mathcal{F}$ as $w(f)=\sum_{x\in f} w(x)$. Then with probability at least $1-\frac{|U|}{b}$ there is a unique $f\in\mathcal{F}$ that has the minimum weight among all sets of $\mathcal{F}$.
\end{lemma}

In our application of the above Lemma, we let $U=E$ and the family $\mathcal{F}$ be the Hamiltonian cycles through $e^*$ in $G$. The Lemma now says that if we set the weight bound $b=2|E|$, say, and set our weights $w((u,v))$ for edges $(u,v)\in E(G)$ uniformly and independently at random from $[1,\ldots,b]$, then with probability at least $\frac{1}{2}$ the Hamiltonian cycle with the smallest weight $w_{\mbox{\tiny min}}$ will be unique. To encode the weights we will use powers of the auxiliary $y$-variable as a factor of each arc's encoding in the matrix, cf. \eqref{eq: M_e}. Since $y^ay^b=y^{a+b}$, the monomial of $y^{w^*}$ in the determinant of $M_{s,t}(x)$ will count the contributions of Hamiltonian cycles in $G$ of weight $w^*$ only.

\begin{lemma}
\label{lem: odd monomial}
If there is a single special Hamiltonian cycle through the edge $e^*=\{s,t\}$ with weight $w^*$, then there is exactly one nonroot to $f$.  
\end{lemma}
\begin{proof}
As in the proof of Lemma~\ref{lem: vanishing}, we know a contributing cycle cover consists of a Hamiltonian cycle through $\{s,t\}$. Moreover, it must pass through the arc $(t,s)$, since the entry corresponding to the oppositely directed arc $(s,t)$ in $M_{s,t}(x)$ is zero. Inserting \eqref{eq: M_e} in \eqref{eq: Leibniz}, we get
\begin{equation}
 \operatorname{det}(M_{s,t}(a))=\sum_{\substack{\sigma\\\{v,\sigma(v)\}\in E}}y^{w(\{t,s\})}(a_s+1)\prod_{v\in V\setminus\{t\}} y^{w(\{v,\sigma(v)\})}(a_v+a_{\sigma(v)}+p_{s,t}((v,\sigma(v))+1)).
\end{equation}
We see that only assignments $a_1,\ldots, a_n$ to the variables $x_1,\ldots,x_n$ that is a solution to the \textsc{Realisable Cycle Cover} problem with $a_s=0$ in the subgraph of $G$ represented by $M_{s,t}(x)$'s nonzero entries, contribute to the determinant's value since each factor in the inner product needs to be nonzero. According to Lemma~\ref{lem: realisable}, $\sigma$ needs to be a Hamiltonian cycle in $G$ through $e^*=\{s,t\}$, and we only filter out the Hamiltonian cycles of weight $w^*$. In particular, for such a cycle, if $a_s=0$ is fixed, all the remaining $a_i$ are forced by following the cycle starting from $s$, i.e, if $h=(s=h_0,\ldots,h_{n-1}=t)$ is a special Hamiltonian cycle, we can set $a_s=a_{h_0}=0$, and $a_{h_{i+1}}=a_{h_i}+p((h_i,h_{i+1}))$ for $i=0,\ldots, n-2$. Hence, if there is a unique special Hamiltonian cycle through $e^*$ in $G$ of weight $w^*$, there is exactly one assignment $a$ that makes a nonzero contribution, meaning that $f(a)\neq 0$ for that assignment and it is zero everywhere else.
\end{proof}

Next, if we guess $w^*=w_{\mbox{\tiny min}}$ from the range $[1,\ldots,bn]$ and pick an $e^*$ on this Hamiltonian cycle, we get both a unique special Hamiltonian cycle and a single nonroot, an odd number of nonroots, with probability at least $1/\operatorname{poly}(n)$, i.e., property P3.

\subsection{The arithmetic circuit construction}
\label{sec: C_f}
We will next describe how we can build an arithmetic circuit $C_f$ for $f$ in \eqref{eq: f}.
Berkowitz~\cite{Berkowitz1984} showed that there is an algebraic circuit that computes a determinant of an $m\times m$ matrix over any commutative ring with unity in $O(m^4)$ ring operations. We first apply this construction for a computation over the polynomial ring $\mathbf{Z}_2[x_1,\ldots,x_n,y]$. I.e., our first arithmetic circuit $C^*$ computes 
\[
\operatorname{det}(M_{s,t}(x)),
\]
a polynomial not only in $x_1,\ldots,x_n$, but also in $y$. We will next see how to get rid of the $y$'s and obtain an arithmetic circuit $C_f$ computing the polynomial $f$ in $\mathbf{Z}_2[x_1,\ldots,x_n]$ with our properties P1--P3. First note that each gate node $\alpha$ in the circuit $C^*$ computes a polynomial $f^{(\alpha)}$ in $x_1,\ldots,x_n,y$ as
\[
\alpha=\sum_{i=0}^{w^*} y^if^{(\alpha)}_i(x_1,\ldots,x_n).
\] 
If there are even higher degree monomials in $y$ we can safely discard them, as both addition and multiplication cannot decrease the degree of a monomial and in the final output we are only interested in the monomial with $y^{w^*}$.
Split each gate node $\alpha$ into $w^*+1$ copies $\alpha_i=y^if^{(\alpha)}_i(x_1,\ldots,x_n)$ for $i=0,\ldots,w^*$ and build a new circuit with gate nodes $\alpha'_i=[y^i]\alpha_i$. Now note that if $\alpha$ was a sum gate node $\alpha=\beta+\gamma$ in the original circuit, then we can set $\alpha'_i=\beta'_i+\gamma'_i$ for all $i$ in the new circuit, and if $\alpha$ was a product gate node $\alpha=\beta \times \gamma$, we can insert the formula $\alpha'_i=\sum_{j=0}^i \beta'_j\times \gamma'_{i-j}$ in the new arithmetic circuit. The size of the new circuit is increased at most by a factor $w^*$ from the old one, and $w^*$ is at most $2|E||V|$ in the original graph's size parameters, a polynomial in $n$. From the original output node $\omega$ in the circuit, we identify $\omega'_{w^*}$ as the output of the new arithmetic circuit and discard all the other outputs and any sub gate nodes that were only needed for their computation. This completes our construction of $C_f$.

\section*{Acknowledgments}
We thank the anonymous reviewers of an earlier manuscript as well as a group of researchers at ITU Copenhagen attending an informal presentation of the result for several comments improving the presentation. 
\bibliographystyle{plainurl}
\bibliography{cn-np-hard}
\end{document}